\documentclass[aps,superscriptaddress,nofootinbib,showpacs,notitlepage, twocolumn]{revtex4-1}
\usepackage{amsfonts}
\usepackage{amsmath}
\usepackage{amsthm}
\usepackage{braket}
\usepackage{graphics}
\usepackage{hyperref}
\usepackage[dvipsnames]{xcolor}
\usepackage{amssymb}
\usepackage{dsfont}
\usepackage{verbatim}
\usepackage{amsmath,amscd}
\usepackage[all,cmtip]{xy}
\usepackage{multirow}

\definecolor{darkblue}{RGB}{0,0,127} 
\definecolor{darkgreen}{RGB}{0,130,80}
\definecolor{darkred}{RGB}{150,10,10}
\hypersetup{
	colorlinks,
	linkcolor=darkblue,
	citecolor=darkgreen,
	filecolor=red,
	urlcolor=blue,
	pdftitle={},
	pdfauthor={}
}

\graphicspath{{./Figures/}}

\usepackage{tikz,pgfplots}\pgfplotsset{compat=newest}
\usetikzlibrary{external,calc,decorations.pathreplacing,decorations.markings,decorations.pathmorphing,arrows.meta,shapes.geometric}
\newlength\figureheight
\newlength\figurewidth

\newcommand{\Z}{\ensuremath{\mathbb{Z}}}

\newcommand{\topo}{\text{TO}}
\newcommand{\R}[1]{{Ref.~\onlinecite{#1}}}
\newcommand{\strop}{\ensuremath{S}}

\newtheorem{lemma}{Lemma}
\newtheorem{claim}{Proposition}

\newcommand{\drawgenerator}[8]{%
\xymatrix@!0{%
& #8 \ar@{-}[ld]\ar@{.}[dd] \ar@{-}[rr] & & #7 \ar@{-}[ld]  \\%
#1 \ar@{-}[rr] \ar@{-}[dd] &  & #2 \ar@{-}[dd] &            \\%
& #6 \ar@{.}[ld] &  & #5 \ar@{-}[uu] \ar@{.}[ll]       \\%
#3 \ar@{-}[rr] &  & #4 \ar@{-}[ru]                       %
}%
}

\newcommand{\plaquette}[4]{
\xymatrix@!0{%
#1 \ar@{-}[r] \ar@{-}[d]  & #2 \ar@{-}[d] 
\\
#3 \ar@{-}[r]  & #4
}}

\usepackage{floatrow}
\usepackage[caption=false,label font={bf,normalsize}]{subfig}
\newcommand{\FFtwo}{{\mathbb{F}_2}}
\newcommand{\FF}{{\mathbb{F}}}
\newcommand{\ZZ}{{\mathbb{Z}}}
\newcommand{\mm}{{\mathfrak{m}}}
\DeclareMathOperator{\coker}{\mathrm{coker}}
\DeclareMathOperator{\rank}{\mathrm{rank}}

\begin{document}

  \title{Compactifying fracton stabilizer models} 
\author{Arpit Dua}
  \affiliation{Department of Physics, Yale University, New Haven, CT 06520-8120, USA}
  \affiliation{Yale Quantum Institute, Yale University, New Haven, CT 06520, USA}

 \author{Dominic~J. Williamson}
  \affiliation{Department of Physics, Yale University, New Haven, CT 06520-8120, USA}
  
  \author{Jeongwan Haah}
  \affiliation{Quantum Architectures and Computation,
Microsoft Research, Redmond, WA 98052, USA}
  
  \author{Meng Cheng}
  \affiliation{Department of Physics, Yale University, New Haven, CT 06520-8120, USA}

  \begin{abstract}
We investigate two dimensional compactifications of three dimensional fractonic stabilizer models. We find the two dimensional topological phases produced as a function of compactification radius for the X-cube model and Haah's cubic code. 
Furthermore, we uncover translation symmetry-enrichment in the compactified cubic code that leads to twisted boundary conditions.  
  \end{abstract}

  \maketitle

  \section{ Introduction}
  \label{sec:intro}

The exploration of topological order in three dimensions led to the surprising discovery of fracton models, whose topological excitations have emergent mobility restrictions~\cite{chamon2005quantum,haah2011local,kim20123d, yoshida2013exotic,haah2013commuting,PhysRevB.92.235136, vijay2016fracton, PhysRevB,PhysRevLett.116.027202,PhysRevB.95.155133,you2018symmetric,PhysRevB.97.134426,Schmitz2019}. A large class of these models can be realized by local commuting projector Hamiltonians, and obey lattice definitions of topological order~\cite{bravyi2010topological}. They are broken up into type-I~\cite{ PhysRevB.81.184303,doi:10.1080/14786435.2011.609152,bravyi2011topological, PhysRevB.95.245126,vijay2017isotropic,vijay2017generalization,PhysRevB.96.165106,PhysRevLett.119.257202,PhysRevB.97.155111,PhysRevB.97.041110,prem2018cage,Bulmash2018} and type-II categories~\cite{ PhysRevLett.107.150504,bravyi2013quantum,Haah2013,haah2014bifurcation} depending on whether or not there are any string operators in the theory. 
Similar phenomena have been observed in higher rank gauge theories~\cite{PhysRevB.95.115139,PhysRevB.96.035119,PhysRevD.96.024051,PhysRevB.96.125151,PhysRevB.96.115102,PhysRevB.97.085116,prem2018pinch,PhysRevB.97.235102,PhysRevLett.120.195301,Pretko2018,PretkoGauge2018,Slagle2018,gromov2017fractional, ma2018fracton,PhysRevB.97.235112,bulmash2018generalized,Williamson2018Fractonic,Gromov2018}. For a review of fracton lattice models and gauge theories see \R{nandkishore2018fractons}.

In this work we explore the properties of fracton models under compactification. Broadly speaking, compactification starts from a $d$-dimensional manifold of the form $M_{d-1}\times S^1$, and passes to the limit where the linear dimension of the $(d-1)$-manifold $M_{d-1}$ is much larger than the size of the circle $S^1$. The $d$-dimensional theory is thus reduced to a $(d-1)$-dimensional one, which is oftentimes easier to analyze.

Previously, compactification has proven a very useful tool to understand topological phases that fall into the framework of topological quantum field theory (TQFT)~\cite{DijkgraafWitten,yetter1993tqft,crane1993categorical,crane1994four,carter1999structures,walker2012,kapustin2013higher,shawnthesis,williamson2016hamiltonian,barenz2016dichromatic,bullivant2016topological,Delcamp2018},  in both two and three dimensions. In 3D,
compactification reduces complicated loop excitations to more familiar point-like anyons, and allows for computation of the so-called three-loop braiding statistics of loop excitations~\cite{jiang2014generalized,Juven2014,wang2014braiding}. 
In fact, many known (3+1)D TQFTs, including abelian twisted Dijkgraaf-Witten gauge theories, can be completely classified according to their behavior under compactification~\cite{wang2015topological}. In two dimensions, compactification was exploited in the classification of symmetry-enriched topological phases (SET)~\cite{PhysRevB.90.045142, Zaletel2017, barkeshli2016reflection}. 

It is now widely recognized that fracton phases are sensitive to more than just topology, with their properties depending on certain geometric structures of the underlying manifold~\cite{ PhysRevB.97.165106, PhysRevB.96.195139,shirley2017fracton,shirley2018FoliatedFracton,shirley2018Fractional,shirley2018universal,shirley2018Foliated,Slagle2018foliated,Tian2018}. In this work we only consider translation-invariant fracton models defined on regular cubic lattices. For the purpose of defining compactification, it is necessary to keep track of the length of the system in the compactified direction $\hat{z}$, which is well-defined when there is translation symmetry along $\hat{z}$. We find that type-I and type-II fracton models behave very differently under compactification. In particular, we show that the fractal mobility constraint of type-II fractons is mapped onto the interplay between translation symmetries  and the topological order in the resulting (2+1)D model.  
 This further provides a method to distinguish different type-II fracton models. 
For example, while Haah's cubic code and the renormalized cubic code B are closely related, as the latter was obtained from real space renormalization of the former, they are nevertheless distinct phases of matter~\cite{haah2014bifurcation} which is reflected in the difference between their translation actions.
Surprisingly, we also find that the compactified models may exhibit subsystem symmetry-enriched topological order~\cite{you2018subsystem,subsystemphaserel,devakul2018fractal,devakul2018universal,kubica2018ungauging} 
leading to spurious contributions to their topological entanglement entropies~\cite{bravyiunpublished,PhysRevB.94.075151,Williamson2018}. 

The paper is laid out as follows: in Section~\ref{compactification} we describe some general features of compactifying 3D stabilizer models, in Section~\ref{sec:tqft} we review the general structure of a compactified (3+1)D TQFT, giving the toric code as an example, 
in Section~\ref{sec:type1} we turn to the compactification of type-I fracton models, with a focus on the X-cube model, 
in Section~\ref{sec:type2} we present the core results of the paper concerning compactification of type-II fracton models, focusing on Haah's cubic code, 
finally we conclude in Section~\ref{sec:conclusion}. 
In the Appendix we present the data we have obtained by compactifying a large range of fracton stabilizer models. 

\section{Compactifying topological codes}
 \label{compactification}
We begin with some general comments on compactifying topologically ordered lattice models. 

Consider a gapped 3D lattice model defined on a spatial manifold $M_2\times S^1$, where we take $S^1$ to be the $\hat{z}$ direction with periodic boundary condition. Taking the limit $L_x,L_y\rightarrow \infty$, while keeping $L_z$ fixed to a finite value produces a 2D gapped lattice model.  We follow the generally held belief that all 2D gapped spin lattice models are described by TQFTs and hence fall into stable phases labeled by an anyon model (modular tensor category) together with a consistent chiral central charge~\cite{Moore1989,kitaev2006anyons}.  We also allow the possibility of ``unstable" phases that are direct sums of the stable phases that will instantaneously decohere into one of the constituents. Physically this corresponds to some kind of spontaneous symmetry breaking. This happens whenever the 3D model supports nontrivial string operators that can wrap the $\hat{z}$ direction. In this case, the string operators wrapping the $S^1$ cycle become local operators in the 2D compactified model. Generally speaking, these operators can acquire finite expectation values, splitting the degeneracy between different sectors labeled by eigenvalues of the string operators. For models with zero correlation length (e.g. stabilizer codes), no such splitting occurs and the resulting ground space is a direct sum of all sectors.

For the infinite cubic lattice, our focus in this paper,  we may specify an arbitrary lattice vector when we compactify a model. Hence we may associate a function $n_{\vec{v}}^{\{ s_i \}}$ to any lattice vector $\vec{v}$ and eigenvalues $\{ s_i \}$ of the  string operators that become local under compactification. This function takes values in the space of 2D topological phases. 

\subsection{Structure theorems for stabilizer codes}

\begin{figure}
\includegraphics[width = 3.5in]{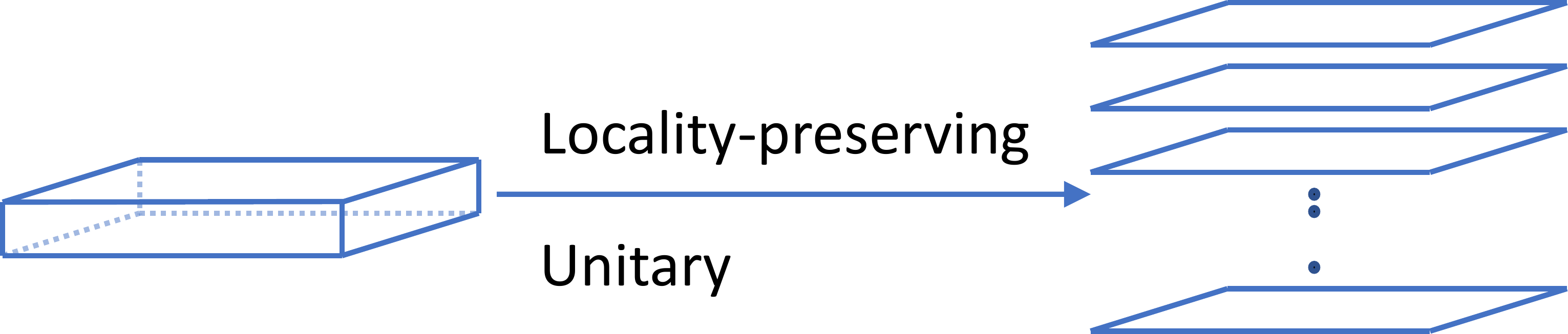}
\caption{Structure theorem~\cite{bombin2014structure,haah2016algebraic,Haah2018a}: a translation invariant topological stabilizer Hamiltonian in 2D is equivalent to copies of toric code and a trivial state via a locality-preserving unitary. This can be used to decompose a compactified 3D model (left) into copies of the 2D toric code (right).  }
\end{figure}

From hereon, we focus on translation invariant stabilizer Hamiltonians with topological order. 
These Hamiltonians are specified by a choice of local Pauli stabilizer generators $h^{(i)}$ which are given, up to a sign, by a tensor product of Pauli $X,Y,$ and $Z$ matrices, where 
\begin{align}
     X = \begin{pmatrix}
        0 & 1  \\
        1 & 0
    \end{pmatrix}
    \, ,
    && 
     Y = \begin{pmatrix}
        0 & -i  \\
        i & 0
    \end{pmatrix}
    \, ,
    &&
    Z = \begin{pmatrix}
        1 & 0  \\
        0 & -1
    \end{pmatrix}
    \, .
\end{align}
The stabilizer generators become the interaction terms 
\begin{align}
    H = \sum_{\vec{v}} (\openone -  h^{(i)}_{\vec{v}}) \, ,
\end{align}
where $h^{(i)}_{\vec{v}}$ indicates a local generator $h^{(i)}$ after translation by a lattice vector $\vec{v}$, in fact we can loosen this translation invariance by including a $\vec{v}$-dependent sign factor. 
We require the Hamiltonian to consist of commuting terms, be frustration free, and satisfy the topological order condition defined in Ref.~\cite{Haah2013}. 
We say that a stabilizer Hamiltonian is CSS~\cite{PhysRevA.54.1098,Steane2551} if each generator $h^{(i)}$ consist exclusively of either Pauli $X$ or $Z$ terms.

Starting with a translation invariant 3D topological stabilizer Hamiltonian, we can compactify it along a lattice vector $\vec{v}$ and fix the eigenvalues $\{ s_i \}$ of any string operators that become local. This leaves us with a translation invariant (up to signs) 2D topological stabilizer Hamiltonian. We can then rely on the existing rigorous classification of such models~\cite{bombin2014structure,Haah2018a} which ensures that they are equivalent, up to a locality preserving Clifford unitary circut, to a finite number of copies of the 2D toric code~\cite{qdouble} and some trivial product state. 
If the stabilizer Hamiltonian is CSS a Clifford local unitary suffices~\cite{haah2016algebraic}. 
This implies that, after fixing out any local degeneracy, the complete topological phase invariant $n_{\vec{v}}^{\{ s_i \}}$ for a compactified 3D topological stabilizer model is simply the number of copies of 2D toric code it is equivalent to. 

There is a similar structure theorem for 1D stabilizer models~\cite{Haah2013} that is useful for calculating the compactification of a 2D stabilizer Hamiltonian. This structure theorem states that any translation invariant 1D stabilizer model is equivalent to copies of the 1D quantum Ising model and some trivial product state.

\subsection{Calculating the 2D topological order}

We now outline how to analyze the 2D stabilizer models obtained via compactification. For the remainder of the paper we consider compactifying along a spatial axis, usually chosen to be $\hat{z}$. 
Technically, we pass from 3D to 2D by grouping all the qubits along a single column in the $\hat{z}$ direction, which is being compactified, onto a single lattice site. Thus, a 3D model with $Q$ qubits per site is mapped to a 2D model with $Q L_z$ qubits per site, where $L_z$ is the compactification radius. Next we fix the eigenvalue of any string operators wrapping the $\hat{z}$ direction by adding a term proportional to them to the Hamiltonian. By the structure theorem we know this 2D stabilizer code is equivalent to a number $n_{L_z}^{\{ s_i \}}$ of toric codes. Hence we can find a basis of $2n_{L_z}^{\{ s_i \}}$ anyonic excitations that generate all others via fusion. These anyons appear at the ends of $2n_{L_z}^{\{ s_i \}}$ string operators, that can be organized into anti-commuting pairs. 
The commutation matrix of these string operators, defined below, determines the $S$ matrix invariant of the relevant topological phase~\cite{haah,ribbons}. Half the rank of the commutation matrix is equal to the number of toric codes and hence in this case it is a complete invariant. It furthermore does not suffer from spurious contributions, due to subsystem symmetries, which may afflict attempts to extract the number of copies of toric code from a topological entanglement entropy calculation~\cite{Williamson2018}.

To calculate the number of anti-commuting string operator pairs in a 2D stabilizer model we consider two overlapping strip-like sub-regions of the 2D lattice, one horizontal and one vertical, as shown in Fig.~\ref{strips}.
We search for Pauli string operators in these subgregions that may only create excitations at their end points. 
These are included in the kernel of the excitation map with closed boundary conditions along the length of the strips, and open boundary conditions at the ends, 
where closed boundary conditions correspond to including stabilizer generators that cross the boundary, and open boundary conditions do not. 

We label the generating set of Pauli string operators found on the horizontal (vertical) strip by $\strop^{h}_i,\, i=1, \dots , N$ $(\strop^{v}_j,\, j=1, \dots, N')$. However, many of these strings may  create trivial anyons in the vacuum sector. To address this we construct the commutation matrix $C$, which is defined elementwise by 
\begin{align}
    C_{i,j} = 
    \begin{cases} 
   0 & \text{if }  [\strop^{h}_i,\strop^{v}_j]=0 \, , \\
   1       & \text{if }  [\strop^{h}_i,\strop^{v}_j] \neq 0 \, .
  \end{cases}
\end{align}
The $\Z_2$-rank of $C$ is the number of independent pairs of anti-commuting string operators. For these string operators to anti-commute the anyons they create must be nontrivial. Hence, half of the rank gives the number of copies of 2D toric code equivalent to the model, as each copy contributes two independent string operator generators. Therefore, if the 2D stabilizer model has been obtained from a compactification specified by $L_z,\,\{s_i \}$, we have
\begin{align}
    n_{L_z}^{\{ s_i \}} = \frac{1}{2} \rank C
    \, .
\end{align}

\subsection{Calculating the action of translation}

\begin{figure}
 \centering
\includegraphics[scale=0.35]{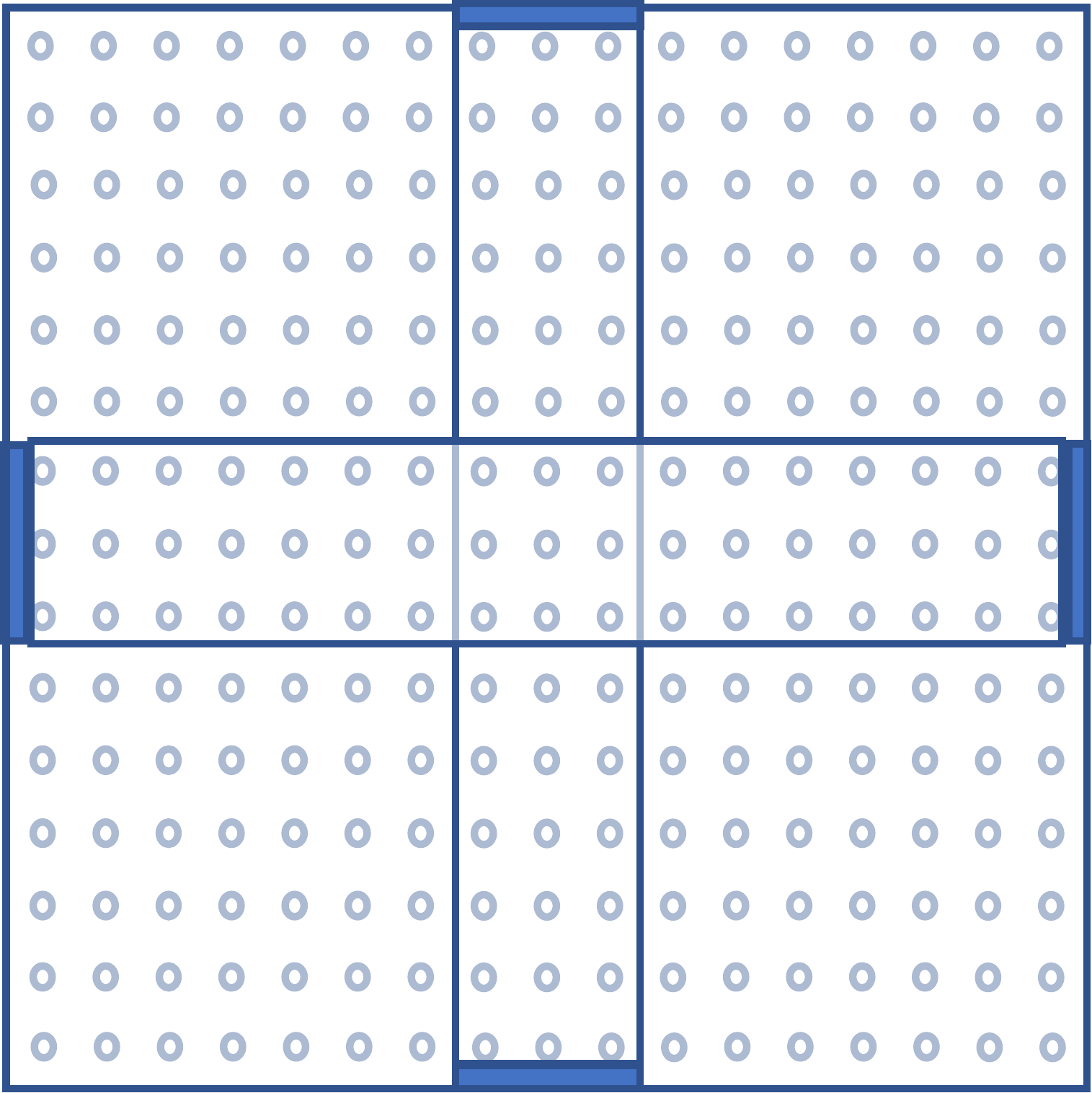}
\caption{
A section of a 2D stabilizer model is shown, with circles depicting the qubits. Two strips are shown, with solid blue bars on the ends denoting open boundary conditions, and thin lines denoting closed boundary conditions. Half the rank of the commutation matrix between operators on the strips gives the number of copies of toric code.
}
\label{strips}
\end{figure}

Calculating the action of translation upon the anyons requires more detailed information about the topological order. We first isolate a basis of string operators that move a generating set of anyons  from those that create topologically trivial excitations. This is achieved by bringing the commutation matrix into Smith normal form
\begin{align}
    U C V = D \, ,
\end{align}
where $U$ and $V$ are invertible (unimodular) matrices and $D$ is the Smith normal form, which is diagonal. 
Since we are working over $\Z_2$, the nonzero entries of $D$ consist of $\rank C$ entries of 1 along the diagonal. 
Hence we can truncate $U$ to its first  $\rank C$ rows, call this $\bar{U}$, and $V$ to its first $\rank C$ columns, call this $\bar{V}$, and we have
\begin{align}
    \bar{U} C \bar{V} = \openone_{\rank C} \, .
\end{align}
The rows of $\bar{U}$ specify products of string operators $\strop^{h}_i$ that give a basis for the horizontal strip: $\bar{\strop}^{h}_i,\, i=1,\dots,\rank C$.  
Similarly the columns of $\bar{V}$ specify a conjugate basis for the vertical strip: $\bar{\strop}^{v}_j,\, j=1,\dots,\rank C$, satisfying
\begin{align}
       [\bar{\strop}^{h}_i,\bar{\strop}^{v}_j]  =
   \begin{cases} 
   0 & \text{if }  i\neq j \, , \\
   2 \bar{\strop}^{h}_i \bar{\strop}^{v}_j       & \text{if }  i=j \, .
  \end{cases}
\end{align}
For CSS code models, we can further ensure that the string operators in our basis move self-bosonic anyons by treating the Pauli $X$ and $Z$ strings separately. We have done so for all the CSS code examples considered in this work. 

The string operator basis we have found allows us to calculate the action of translation in the horizontal direction by simply shifting the vertical string operators over one site, $T_x^{-1} \bar\strop^{v}_j T_x$, and similarly for the vertical translations $T_y^{-1} \bar\strop^{h}_i T_y$. 
This is because we can use the modularity of anyon braiding to identify the shifted vertical strings with the original basis via their commutation relations with the unshifted horizontal strings. 
To implement this we form a new commutation matrix given by 
\begin{align}
    (C_x)_{i,j} = 
  \begin{cases} 
   0 & \text{if }  [\bar\strop^{h}_i,T_x^{-1} \bar\strop^{v}_j T_x]=0 \, , \\
   1       & \text{if }  [\bar\strop^{h}_i,T_x^{-1} \bar\strop^{v}_j T_x] \neq 0 \, .
  \end{cases}
\end{align}
Taking the Smith decomposition we find 
\begin{align}
U_x C_x V_x &= \openone 
\nonumber \\
&=  C_x V_x U_x  \, ,
\end{align}
hence the change of basis matrix for the vertical strings $V_x U_x$ gives the generator of the translation action on the anyons $\bar{T}_x$. 
Similar reasoning yields $\bar{T}_y = V_y U_y $.

The action of translation on the anyons is fully specified by  $\bar{T}_x$ and $\bar{T}_y$ in the examples we have considered. This is because these matrices describe the permutation of anyons by translation, and no symmetry fractionalization occurs in the examples considered. 
However, it can be hard to interpret these matrices directly. 
In the examples below we have focused on a simple invariant of the translation action:  $O_i$, the order of the permutation. That is, the smallest nonzero solution to $\bar{T}_i^{O_i}=\openone$. These orders give the dimensions of the smallest unit cell $O_x \times O_y$ one must take such that the coarse grained translation does not permute anyons. Moreover, for an $L_x \times L_y$ system with periodic boundary conditions, if an $O_i$ does not divide $L_i$ the boundary conditions will be twisted. This leads to a reduction in the ground state degeneracy, compared to the untwisted value $4^{n_{L_z}^{\{ s_i\}}}$. 

Under compactification, the remaining 3D translation symmetry along the $\hat{z}$ direction is mapped to a $\mathbb{Z}_{L_z}$ on-site permutation symmetry. Following the method above we can work out the action $\bar{T}_z$ of this symmetry on the anyons in the resulting 2D model, and derive properties such as its order $O_z$.

 \section{Compactification of TQFTs} 
 \label{sec:tqft}
 
In this section we review the compactification of 3D topological phases that are described by TQFTs. It is generally believed, on physical grounds, that TQFTs in (3+1)D support two types of excitations at low energy: particles and loops, both of which are mobile. These excitations can be created and moved by appropriate topological string and membrane operators, which can be locally deformed as long as their boundary is kept fixed. By definition, topological invariance ensures that the compactified (2+1)D theory will also be a TQFT and that all choices of compactification direction and radius leading to the same theory. The resulting (2+1)D theory will generally not be stable, but will correspond to a \emph{cat state} consisting of an unstable sum of multiple different stable TQFTs. This is due to topologically nontrivial string operators wrapping the compactified direction, which become local operators in the compactified theory. We expect (under a modularity assumption) that such compactified string operators can be used to project onto a definite flux, or string,  excitation threading the compactified $S^1$. Topologically nontrivial membrane operators on a plane orthogonal to the compactified direction become global symmetries of the resulting (2+1)D model, and the compactified string operators split the symmetry breaking degeneracy. On the other hand, membrane operators wrapping the compactification direction become string operators for flux loops wrapping the compactified direction, which are mapped to point-like particles. Hence, the topological order after compactification is given by
\begin{align}
\topo^\text{comp}_{3\text{D}} \cong \bigoplus_{s} \topo_{2\text{D}}^{\{s\}} .
\end{align}
where $\{s\}$ denotes the eigenvalues of the string operators. 
Assuming the original set of string excitations and dual string operators was complete, each of the resulting $2D$ theories should be stable. 
Hence $n_{\vec{v}}^{\{ s\} } = \topo_{2\text{D}}^{ \{ s \} } $ for arbitrary $\vec{v}$.

In the most general setting, it can be quite difficult to calculate the resulting (2+1)D topological orders $ \topo_{2\text{D}}^{\{s\}} $. However, the answer is known for all twisted Dijkgraaf-Witten theories~\cite{DijkgraafWitten} (twisted quantum double models). 
This reduction has proven useful for the calculation of many properties of (3+1)D TQFTs, including the classification of loop defects, and the computation of 3-loop braiding statistics, and generalizations thereof, which completely classify abelian Dijkgraaf-Witten theories~\cite{jiang2014generalized,Juven2014,wang2014braiding,wang2015topological}.

\subsection{2D toric code}
As a warm up we consider compactifying the 2D toric code. This example captures many of the essential features of compactifying a TQFT, and is relevant for the type-I example to follow. 

The 2D Toric code model is commonly defined with one qubit per edge of a square lattice. Here we group a pair of qubits from orthogonal, adjacent edges onto each site. The local stabilizer generators in the Hamiltonian are then given by 
 \begin{align}
 \label{eq:tc}
 \begin{array}{c}
 \plaquette{ZI}{II}{ZZ}{IZ}
 \quad
 \plaquette{XI}{XX}{II}{IX}
 \end{array}
 \, .
\end{align}
The logical operators on a torus are given by anti-commuting pairs of string operators $\bar{X}_{\hat{i}},\bar{Z}_{\hat{j}}$, where ${\hat{i}\neq\hat{j} \in \{\hat{x},\hat{y}\}}$. One pair of representative logical operators is given by 
\begin{align}
\bar{X}_{\hat{x}} = \prod_x IX_{(x,0)} \, ,
&&
\bar{Z}_{\hat{y}} = \prod_y IZ_{(0,y)}
\, ,
\end{align}
and the other is defined similarly. 

Under compactification along the $\hat{y}$ direction we obtain two copies of the 1D quantum Ising model, up to local unitary. 

To see this, we make use of the structure theorem for translation invariant stabilizer codes in one dimension. This reduces the problem to counting the number of local-nonlocal pairs of logical operators in the compactified model. Since the compactification of 2D Toric code along one direction maps two out of the four logical string operators to local logical operators, the topological ground space degeneracy reduces to an unstable symmetry breaking degeneracy. The remaining two logical operators map to global spin-flip symmetries for two copies of the Ising model, consistent with the fourfold ground state degeneracy. 

\subsection{3D toric code}
For the next example, we consider the 3D toric code. Similar to the 2D toric code, we have grouped three qubits onto each lattice site, and the stabilizer generators are given by 
\begin{align}
\begin{array}{c}
\drawgenerator{IXI}{XXX}{}{IIX}{}{}{XII}{}
\quad
\drawgenerator{}{}{}{}{IIZ}{IZZ}{}{IZI}
\quad
\\
\drawgenerator{}{}{IIZ}{}{}{ZIZ}{}{ZII}
\quad
\drawgenerator{}{}{IZI}{}{ZII}{ZZI}{}{}
\end{array}
\, .
\end{align}	
Representative logical operators on a torus are generated by three anti-commuting membrane-string operator pairs $\bar{X}_{\hat{i}},\bar{Z}_{\hat{i}}$, where $\hat{i}=\hat{x},\hat{y},\hat{z}$. 
A representative pair is given by 
\begin{align}
\bar{X}_{\hat{x}}& = \prod_{y,z} XII_{(0,y,z)}\, ,
&
\bar{Z}_{\hat{x}}& =  \prod_{x} ZII_{(x,0,0)}\, ,
\end{align}
and similarly for $\hat{y}$ and $\hat{z}$. 

Compactifying the $\hat{z}$ direction makes $\bar{Z}_{\hat{z}}$ into a local logical operator that anti-commutes with a global symmetry given by $\bar{X}_{\hat{z}}$. 
The compactification further maps $\bar{X}_{\hat{x}},\bar{X}_{\hat{y}}$ into nontrivial string operators, while $\bar{Z}_{\hat{x}},\bar{Z}_{\hat{y}}$ remain nontrivial string operators, and their commutation relations are preserved. 
Moreover, these generate all the logical operators, up to stabilizers. 
To make a precise identification with copies of toric code we need to be slightly more careful, and study the superselection sectors. We find a basis of two inequivalent superselection sectors for a fixed eigenvalue of $\bar{Z}_{\hat{z}}$. These are represented by: a compactified string excitation that consists of $Z$-stabilizer violations wrapping the $\hat{\hat{z}}$ direction, and a single $X$-stabilizer violation. This, combined with the anti-commutation of the aforementioned string operators that move these anyons, identifies the resulting (2+1)D phase as equivalent to a direct sum of two toric codes 
\begin{align}
\text{TC}_{3\text{D}}^{\text{comp}} \cong \text{TC}_{2\text{D}} \oplus \text{TC}_{2\text{D}} \, .
\end{align}

For this example, and the 2D example above, translation acts trivially.

\section{Compactification of Type-I fracton models}
\label{sec:type1}
In this section we consider the compactification of type-I fracton models, treating the X-cube model as a prototypical example.  
According to their definition~\cite{vijay2016fracton}, type-I fracton models feature some completely immobile point-like excitations (i.e. fractons) alongside other excitations that are mobile in one or more directions. This includes lineons which can move only in one direction, and planons which can move in two directions. Moreover, the fractons should not be composites of mobile excitations that can move in different directions.
In type-I models there can be an extensive number of inequivalent string operators in the direction of compactification that, after being mapped to local operators, leading to an extensive number of symmetry breaking sectors. 
 
\subsection{X-cube model}
The stabilizer generators for the X-cube model~\cite{vijay2016fracton} are given by 
\begin{align}
\begin{array}{c}
\drawgenerator{IXI}{}{IXX}{IIX}{XIX}{XXX}{XII}{XXI}
\quad
\drawgenerator{}{}{}{}{IIZ}{}{IZZ}{IZI}
\quad
\\
\drawgenerator{ZIZ}{}{IIZ}{}{}{}{}{ZII}
\quad
\drawgenerator{}{}{IZI}{ZZI}{ZII}{}{}{}
\end{array}
\, .
\end{align}
 Here we have grouped 3 qubits onto each vertex, similar to the 3D toric code example. 
 We consider this model on an $L_x \times L_y \times L_z$ cuboid with periodic boundary conditions. 
To enumerate the logical operators let us first define anti-commuting pairs  $\bar{X}^{\hat{i}}_{\hat{k},\ell},\bar{Z}^{\hat{j}}_{\hat{k},\ell}$, where $\hat{i}\neq \hat{j} \neq \hat{k} \in \{ \hat{x},\hat{y} ,\hat{z} \} $ and $\ell = 0,\dots,L_k-1,$ along intersecting pairs of non-contractible loops as follows
\begin{align}
\bar{X}^{\hat{x}}_{\hat{z},\ell} = \prod_x XII_{(x,0,\ell)} \, ,
&&
\bar{Z}^{\hat{y}}_{\hat{z},\ell} = \prod_y ZII_{(0,y,\ell)}\, ,
\end{align}
and similarly for other permutations of $\hat{x},\hat{y},\hat{z}$. 
This set of operators is over-complete due to the relations
\begin{align}
    \prod_{\ell} \bar{X}^{\hat{i}}_{\hat{k},\ell} = \prod_{\ell} \bar{X}^{\hat{k}}_{\hat{i},\ell} 
    \, ,
    &&
    \bar{Z}^{\hat{i}}_{\hat{j},0} = \bar{Z}^{\hat{i}}_{\hat{k},0}
    \, .
\end{align}
Removing $\bar{X}^{\hat{x}}_{\hat{z},0},\bar{X}^{\hat{y}}_{\hat{z},0},\bar{X}^{\hat{y}}_{\hat{x},0},\bar{Z}^{\hat{x}}_{\hat{z},0},\bar{Z}^{\hat{y}}_{\hat{z},0},\bar{Z}^{\hat{z}}_{\hat{x},0},$ and replacing
\begin{align}
\bar{Z}^{\hat{x}}_{\hat{z},\ell} &\mapsto \bar{Z}^{\hat{x}}_{\hat{z},0}\bar{Z}^{\hat{x}}_{\hat{z},\ell}
\, ,
&&
\bar{Z}^{\hat{y}}_{\hat{z},\ell} \mapsto 
\bar{Z}^{\hat{y}}_{\hat{z},0}\bar{Z}^{\hat{y}}_{\hat{z},\ell}
\, ,
\end{align}
where $\ell>0$, and
\begin{align}
\bar{X}^{\hat{x}}_{\hat{y},0} \mapsto 
\bar{X}^{\hat{x}}_{\hat{y},0}
\prod_{i=1}^{L_x}  \bar{X}^{\hat{y}}_{\hat{x},i}
    \, ,
\end{align}
we arrive at a complete set of $(2L_x+2L_y+2L_z-3)$ logical operator pairs. 
We remark that this basis has been chosen such that the $2(L_z-1)$ logical operator pairs with a $\hat{z}$ subscript are deformable in the $xy$ plane.

Under compactification along the $\hat{z}$ direction all logical operators with a $\hat{z}$ superscript become local. 
Hence $(L_x+L_y-1)$ logical $\bar{Z}$ operators become local, while their anti-commuting $\bar{X}$ partners are mapped to rigid subsystem symmetries. Similarly $(L_x+L_y)$ logical $\bar{X}$ operators become local, while their anti-commuting logical $\bar{Z}$ partners become linear subsystem symmetries. 
After fixing out the extensive, local degeneracy due to subsystem symmetry breaking, by fixing eigenvalues $\{s_i\}$ of the logical operators with a $\hat{z}$ superscript,  we are left with $(2L_z-2)$ logical operator pairs that are deformable in 2D. These correspond to the string operators of $(L_z-1)$ copies of Toric code wrapping the torus. Hence, each symmetry breaking sector of the compactified X-cube model is local-unitary equivalent to $(L_z-1)$ copies of toric code, where $L_z$ is the compactification radius. That is 
\begin{align}
n_{L_z}^{\{s_i\}} = \text{TC}_{2\text{D}}^{\otimes (L_z -1)} \, .
\end{align}

We have verified these facts numerically using the anti-commuting string operator calculation explained in section~\ref{compactification}. 

In this case the translations $\bar{T}_x,\bar{T}_y,$ act trivially on the copies of toric code, but nontrivially on the symmetry breaking sectors. 
On the other hand, $\bar{T}_z$ acts nontrivially on both the copies of toric code and the symmetry breaking sectors. 

The results of this compactification provide several signatures that differentiate the X-cube model from a decoupled stack of 2D toric code layers. The clearest signature is in the constant correction $-1$ to the number of toric codes occurring as $L_z$ is increased. This is consistent with the fact that one can disentangle a copy of 2D toric code from two layers of the X-cube model via a local unitary~\cite{shirley2018Fractional}. 
From the 2D toric code example in section~\ref{sec:tqft} we also see that the pattern of subsystem symmetry breaking in compactified X-cube is subtly different to that occurring in compactified layers of 2D toric code. 
It would be interesting to use similar signatures from compactification to identify and classify nontrivial type-I models in the future.

\section{ Compactification of type-II fracton models}
\label{sec:type2}

In this section we turn our attention to type-II fracton models, focusing on the canonical example: the cubic code. 
By their definition~\cite{vijay2016fracton}, type-II fracton models do not have any string operators, and thus \emph{all} nontrivial particle excitations are immobile. In such models, isolated excitations can be created at the corners of fractal operators. For example, in the cubic code~\cite{haah2011local}, a fractal operator exists with support on a Sierpinski tetrahedron.

Since there are no string operators in a type-II model, no local logical operators can occur after compactification. Hence the resulting 2D models will be topologically ordered. 
The compactification of a translation invariant type-II stabilizer model in the $\hat{z}$ direction is then simply characterized by the number of toric codes it is equivalent to as a function of length, $n_{L_z}$.

The quantity $n_{L_z}$ has to be calculated on a model-to-model basis, but it does satisfy a simple bound that depends upon the locality of the stabilizer generators. 
Consider compactifying a type-II stabilizer model with generators that act on vertices adjacent to individual cubes, one can show that ${n_{L_z} \leq 2 L_z}$. 
This is achieved by considering the topological entanglement entropy (TEE) computed via the Kitaev-Preskill prescription of tripartite information $I(A,B,C),$ given by 
\begin{align}
    S_A + S_B + S_C - S_{AB} - S_{BC} - S_{AC} + S_{ABC} \, , \nonumber
\end{align}
on the regions shown in Fig.\ref{fig:ABC}. 
More generally for any 2D translation invariant topological stabilizer model, with $Q$ qubits per site and $Q$ independent generators per plaquette which only act on the sites at the corners, we have $n_{L_z} \leq Q$. 
In fact, when there are no relations amongst the generators, and the plaquette stabilizers they generate act nontrivially on all corners, we find saturation $Q=-I(A,B,C)$. 
One would expect this to determine $n_{L_z}$, as it was argued in \R{KitaevPreskill} that $n_{L_z}=-I(A,B,C)$, however it has recently been realized~\cite{bravyiunpublished,PhysRevB.94.075151,Williamson2018} that spurious contributions to the TEE may cause $n_{L_z} \leq -I(A,B,C)$. 
We find that this behaviour occurs for the cubic code example below. 
For this reason we use the more careful approach described in section~\ref{compactification} to calculate $n_{L_z}$ accurately. 

\begin{figure}
\centering
\includegraphics[scale=1.25]{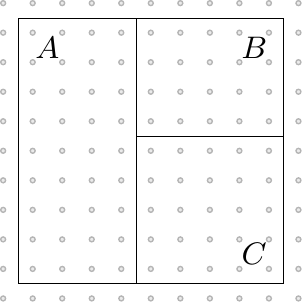}
\caption{Regions used to calculate the topological entanglement entropy.}
\label{fig:ABC}
\end{figure}

The anyons in a compactified type-II model descend from immobile fracton excitations in 3D, raising the question: how does the fractonic immobility manifest in the 2D compactified model? 
We find that the immobility is reflected in the 2D model as a nontrivial enrichment of the topological order by on-site and translation symmetries. 
These symmetries descend from the translation symmetries of the 3D model in the $\hat{z}$, and $\hat{x},\hat{y}$ directions, respectively. 
The immobility of fractons in 3D implies that translates of a nontrivial excitation are topologically distinct and therefore a pair of such excitations cannot be created by a string operator. In other words, the action of translation symmetry must change the topological superselection sectors of the fracton excitations. Descending to 2D, via compactification, this implies that translations permute the anyon types. 
In contrast to the 3D type-II model, the action of translation on its 2D compactification must have a finite order. This can be understood as originating from some operators in 3D that wrap the compactified direction and pairwise create a nontrivial topological excitation and some translate of it. Such an operator would not exist purely in the 3D bulk without the compactified boundary condition. For example, let $O_x$ denote the order of $\bar{T}_x$ in the compactified model, then some string operator in 2D can create any excitation along with its translation by $T_x^{O_x}$. 
We remark that for CSS codes, the action of translation is restricted to permuting excitations of the $X$ stabilizers amongst themselves, and similarly for the $Z$ stabilizers. 

\subsection{Cubic code}

The main example we consider is the cubic code which is known to be a type-II stabilizer model~\cite{haah2011local}, i.e. it has no string operators and all nontrivial excitations are immobile. 
The generators of the stabilizer group are given by 
\begin{align}
\begin{array}{c}
\drawgenerator{XI}{II}{IX}{XI}{IX}{XX}{XI}{IX}
\quad
\drawgenerator{ZI}{ZZ}{IZ}{ZI}{IZ}{II}{ZI}{IZ}
\end{array}
\end{align}
and their translations. 
See Fig.~\ref{excitation_patterns} for the patterns of excitations that are created on the dual cubic lattice by local operators. 
The remainder of this section focuses on the compactification of the above model.

\begin{center}
\begin{table}
\vspace{4mm}
\begin{tabular}{|c|c|c|c|}
\hline 
$L_{z}$ & $n_{L_z}= R / 2 $ & $O_{x}=O_{y}$ & $O_{z}$ 
\tabularnewline
\hline 
\hline 
1 & 2 & 2 & 1  \tabularnewline
\hline 
2 & 4  & 4 & 2 \tabularnewline
\hline 
3 & 2  & 2 & 1 \tabularnewline
\hline 
4 & 8  & 8 & 4 \tabularnewline
\hline 
5 & 10  & 30 & 5 \tabularnewline
\hline 
6 & 4 & 4 & 2  \tabularnewline
\hline 
7 & 14 & 126 & 7  \tabularnewline
\hline 
8 & 16 & 16 & 8  \tabularnewline
\hline 
9 & 14 & 126 & 9 \tabularnewline
\hline 
10 & 20  & 60 & 10 \tabularnewline
\hline 
11 & 22 & 682 & 11 \tabularnewline
\hline 
12 & 8 & 8 & 4 \tabularnewline
\hline 
13 & 26 & 2730 & 13 \tabularnewline
\hline 
14 & 28 & 252 & 14 \tabularnewline
\hline 
15 & 26 & 30 &15  \tabularnewline
\hline 
16 & 32 & 32 &16  \tabularnewline
\hline 
17  & 34 & 510 & 17\tabularnewline
\hline 
18 & 28 & 252 & 18 \tabularnewline
\hline 
19 & 38 & 19418 &19  \tabularnewline
\hline 
20 & 40 & 120 &20  \tabularnewline
\hline 
21 & 38 & 126 & 21 \tabularnewline
\hline 
22 & 44 & 1364 & 22 \tabularnewline
\hline 
\end{tabular}
\caption{As a function of compactification radius $L_z$ we list: $R/2$ half the rank of the commutation matrix $C$,  which equals the number of copies of toric code $n_{L_z}$, for the compactified cubic code. The order  $O_x=O_y$ of the translation action $\bar{T}_x,\bar{T}_y$. The order $O_z$ of the on-site symmetry $\bar{T}_z$.}
\label{table_copies}
\end{table}
\par\end{center}

\subsubsection{Compactified Topological Order}
We consider a family of 2D models produced by compactifying the cubic code, following the procedure in section~\ref{compactification}. 
By the structure theorem presented there, the compactified cubic code models are unitarily equivalent to some number $n_{L_z}$ copies of toric code. 
We have numerically calculated $n_{L_z}$ for a range of values $L_z$ by finding pairs of anti-commuting string operators as explained in section~\ref{compactification}. These results are presented in table~\ref{table_copies}. 

The numerical scaling of $n_{L_z}$ fits a simple formula for $L_z=2^i \ell$, where $2 \not | \, \ell$, given by 
\begin{align}
\label{eq:ntc}
   \boxed{ 
   n_{L_z} = \begin{cases}
    2 L_z & \text{if } 3 \not | L_z
    \, ,
    \\
    2 L_z - 2^{i+2} & \text{if } 3 \ | L_z 
    \, .
    \end{cases}
    }
\end{align}
In the next subsection we prove that the above formula in fact holds for all $L_z$ explicitly, using the framework of polynomial rings~\cite{haah2013commuting}.  
This scaling is substantially simpler than the scaling of the ground space degeneracy of the cubic code on an $L \times L \times L$ 3D torus~\cite{Haah2013}. 
One may then ask: How is it possible that the ground space degeneracy of $n_{L_z}$ copies of toric code on an $L_x \times L_y$ torus matches the ground space degeneracy of the cubic code on an ${L_x \times L_y \times L_z}$ 3D torus? 
The answer lies in the nontrivial action of translation symmetry on the $n_{L_z}$ toric codes contained in the compactified cubic code. Further details are given in the following subsections.

To build intuition for the scaling of $n_{L_z}$ we first notice that an adjacent pair of plaquettes, each having $2L_z$ stabilizers, can support $2^{4 L_z}$ excitations, including vacuum. If these excitations are all topologically distinct, i.e. there are no nontrivial string operators confined to the pair of plaquettes, then the pair of plaquettes can support all the charge sectors of $2 L_z$ copies of toric code. If we suppose that the pair of plaquettes still support all sectors when they also support nontrivial string operators, we can find the topologically inequivalent sectors by modding out the relations introduced by those string operators. 

The simplest case with $n_{L_z} \leq 2 L_z$ occurs when $n_3=2$, which was studied in detail in \R{Williamson2018}. 
To provide an explanation we direct the reader's attention to Fig.~\ref{excitation_patterns}~(e) which shows an excitation pattern created by a local operator. For compactified boundary conditions with $L_z=3$ the operator in Fig.~\ref{excitation_patterns}~(e) allows an excitation to hop between stabilizers related by a $\hat{z}$ translation. Modding out by these relations we find $2^{4}$ inequivalent sectors on the pair of plaquettes. 
Furthermore, combining the operators in Fig.~\ref{excitation_patterns}~(e) with (a) and then (b) we find diagonal hopping operators for each type of excitation on any plaquette. Hence we find a model equivalent to two copies of toric code with $\bar{T}_x,\bar{T}_y$ acting as layer swap and $\bar{T}_z$ acting trivially. This agrees with the numerical results. 

More generally, whenever $3 | L_z$ one can find string operators on a column by taking products of operators similar to that shown in Fig.~\ref{excitation_patterns}~(e). 
For ${2 \not|  L_z}$  there are again two independent string operators on each plaquette for each type of stabilizer. 
Due to the bifurcating real-space renormalization group flow of the cubic code demonstrated in Ref.~\cite{haah2014bifurcation} the number of toric codes $n_{L_z}$ doubles as we double $L_z$, i.e. ${n_{2 L_z} = 2 n_{L_z}}$, which accounts for the value of $n_{L_z}$ in the remaining cases of Eq.~\eqref{eq:ntc} in which $2 | L_z$.

\begin{figure}
\centering
\sidesubfloat[]{\includegraphics[scale=0.32]{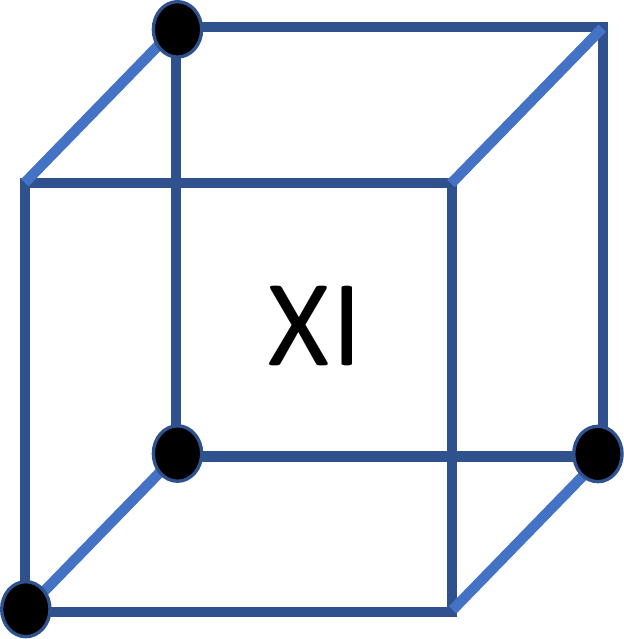}} 
\sidesubfloat[]{\includegraphics[scale=0.32]{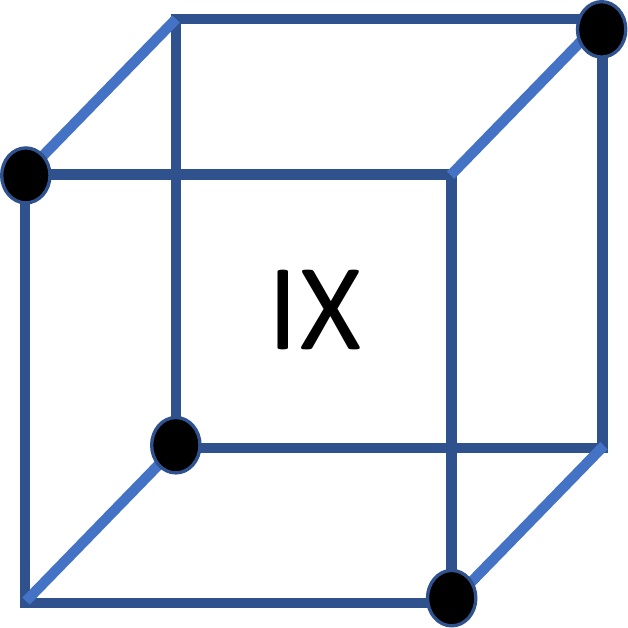}}
\sidesubfloat[]{\includegraphics[scale=0.32]{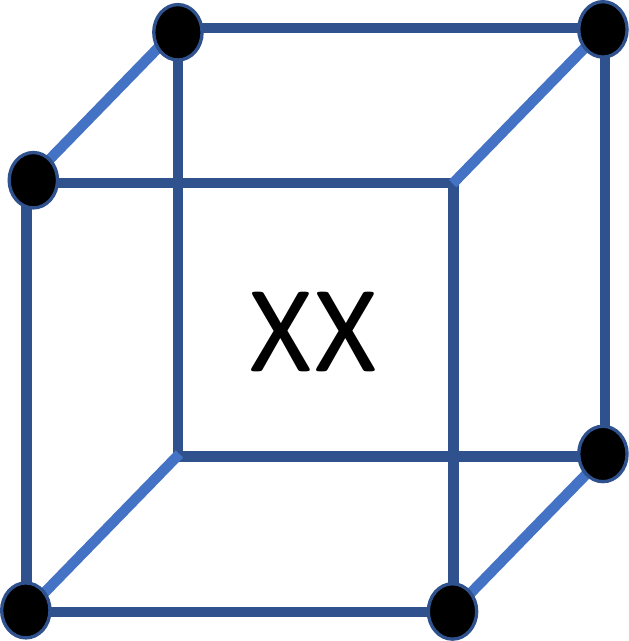}}\\
\sidesubfloat[]{\includegraphics[scale=0.32]{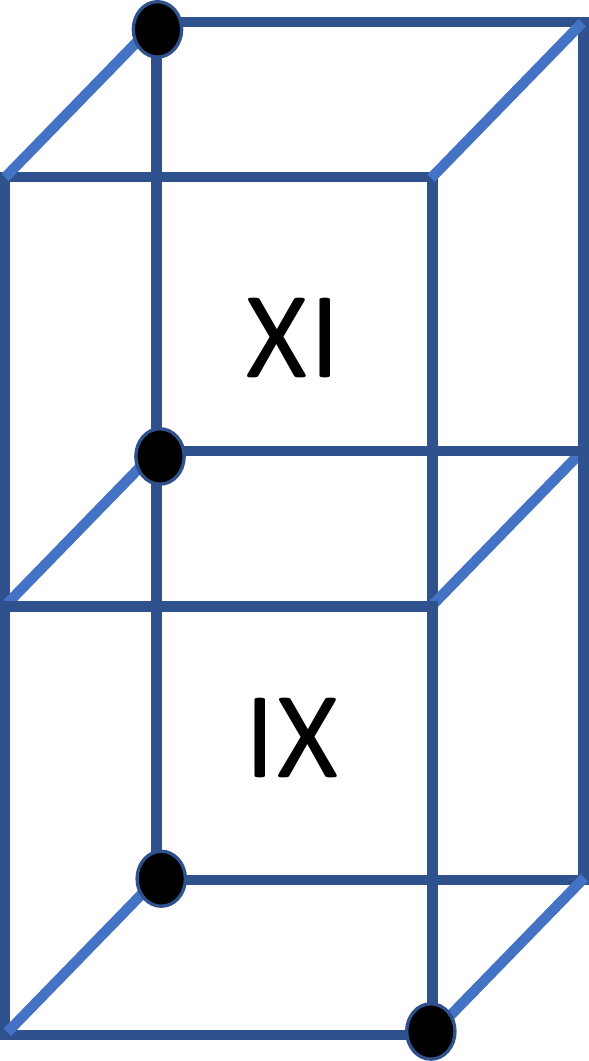}}\hspace{0.1mm}
\sidesubfloat[]{\includegraphics[scale=0.32]{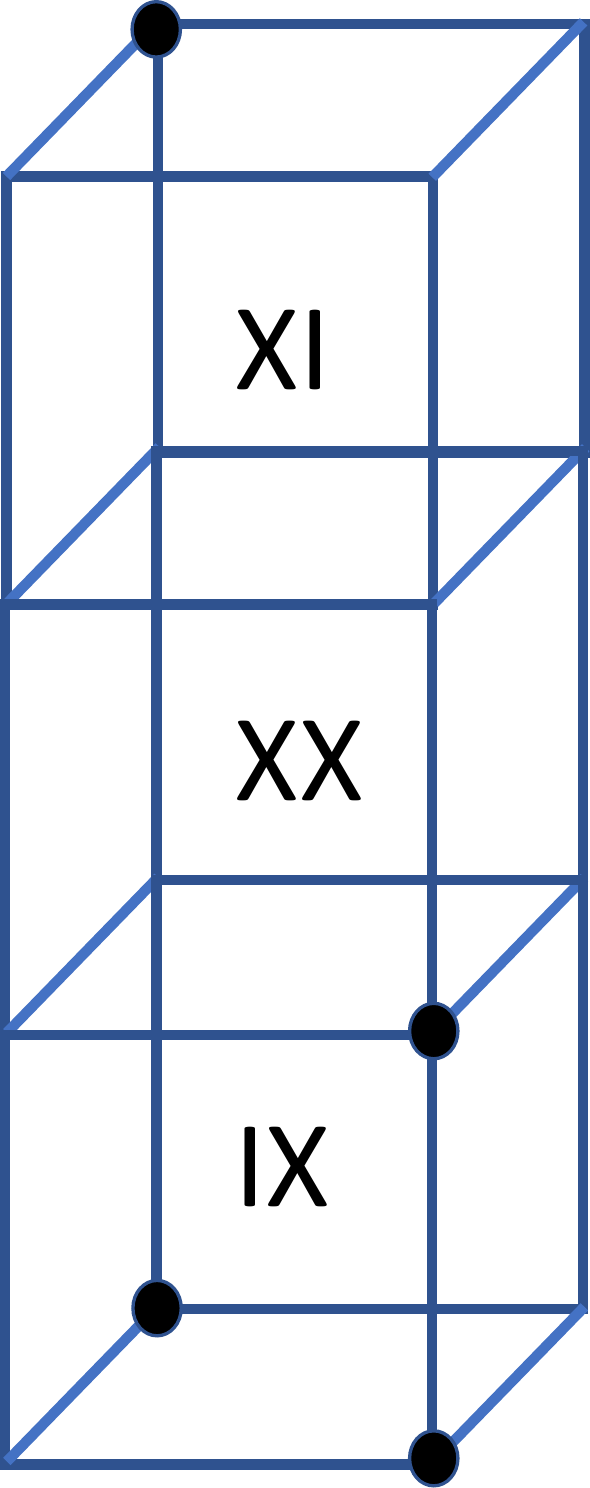}}\\ 
\caption{Excitation patterns of stabilizers located on the vertices of the dual cubic lattice, which are created by local operators on pairs of qubits inside each cube: (a) XI, (b) IX, (c) XX, (d) IX-XI along $\hat{z}$, (e) IX-XX-XI along $\hat{z}$.}
\label{excitation_patterns}
\vspace{-0.7ex} 
\end{figure}

\subsubsection{Copies of Toric Code from the polynomial framework}
In this subsection we compute the number of copies of the 2D toric codes in the cubic code compactified in $z$-direction
by a periodic boundary condition of linear size $L_z$. 
We employ the framework of polynomial rings developed in \R{haah2013commuting}, for the sake of brevity we do not introduce the formalism here. 
The no-strings rule of the cubic code implies that the compactified code is exact.
By the classification theorem of exact CSS codes,
the number $n_{L_z}$ of embedded toric codes is half the dimension of 
the torsion submodule of $\coker \epsilon$~\cite[V.13]{haah2016algebraic}.
Therefore, our starting point is the formula
\begin{align}
&n_{L_z} = \dim_\FFtwo \frac{\FFtwo[x^\pm,y^\pm,z^\pm]}{( 1+x+y+z,1+xy+yz+zx, z^L_z + 1 )} .
\end{align}
As always, we lift the coefficient field to its algebraic closure $\FF \supset \FFtwo$.
The ring in this equation has a finite $\FF$-dimension.
So, it is isomorphic to the finite direct sum of localizations at maximal ideals 
$(x - a-1, y - b-1, z - t -1)$ where $a,b,t \in \FF \setminus \{ 1 \}$ are arbitrary;
the constant $1$ is inserted for convenience in notation later.

Put $L = 2^i \ell$ where $\ell$ is odd.
In order for the localized ring not to vanish,
the constants $a,b,t$ has to be chosen to ``satisfy'' the ideal.
In particular, $a+b + t = 0$, $ab + bt+ ta = 0$ and $(1+t)^\ell + 1 = 0$.
So, $a = \omega t$ and $b = \omega^2 t$ 
where $\omega$ is either one of the two third roots of unity.
Since $a,b,t \neq 1$, the constant $t$ is further restricted so that $t \neq 1,\omega, \omega^2$.
Put $\mm = (u - a, v- b)$ where $a \neq 1$, $b \neq 1$ and $a+b = t \neq 1$.
We can make substitutions $x \to u +1, y \to v + 1, z \to 1+x+y=1+u+v$
so that 
\[
n_{L_z} = \sum_\mm \dim_\FF \FF[u,v]_\mm / \mathfrak J_\mm
\]
where
\begin{align}
\mathfrak J_\mm 
&= ( u^2 + uv + v^2, (1+u+v)^L +1 )_\mm \nonumber\\
&= ( u^2 + uv + v^2, u^{2^i} + v^{2^i} + t^{2^i} )_\mm \\
&= ( u^2 + uv + v^2, uv^{2^i -1} + (n+1)v^{2^i} + t^{2^i} )_\mm. \nonumber
\end{align}
The second line follows from the first because the equation $z^{L_z} +1$ factorizes into linear factors
raised to power $2^i$, but the localization at $\mm$ singles out only one factor.
The third line uses $u^{2^i} = uv ^{2^i-1} + i v^{2^i} \mod u^2 + uv + v^2$
for any integer $ i \ge 0$.

Suppose $t \neq 0$. Then, $\mathfrak J_\mm$ is equal to
$(u - \omega^2 v, (\omega^2 + 1 + i) v^{2^i} + t^{2^i} )_\mm$,
so such $\mm$ contributes to the sum by $2^i$.
If $\ell$ is not a multiple of $3$, then there are $2(\ell - 1)$ such $\mm$.
If $\ell$ is a multiple of $3$, then there are $2(\ell - 3)$ such $\mm$
because we have to exclude the third roots of unity.

Suppose $t = 0$. Then, $a = b = 0$, i.e., $\mm = (u,v)$.
Then,
\[
\mathfrak J_\mm = ( u^2 + uv + v^2, uv^{2^i-1} + (i+1)v^{2^i} )_{(u,v)}.
\]
Regardless of whether $i$ is even or odd,
one can show that $\{ u^2 + uv + v^2, uv^{2^i -1} + (i+1)v^{2^i}, v^{2^i +1} \}$
is a Gr\"obner basis for $\mathfrak J$ in the degree term order.
We see that $(2^i + 1) + (2^i -1) = 2^{i+1}$ contributes to the sum.

Therefore,
\begin{align}
n_{L_z} = 
\begin{cases}
2^i \cdot 2(\ell -1) + 2^{i+1} = 2L_z & \text{ if } 3 \not| L_z , \\
2^i \cdot 2(\ell -3) + 2^{i+1} = 2L_z - 2^{i+2} & \text{ if } 3\ | L_z.
\end{cases}
\end{align}

\subsubsection{Translation symmetry-enrichment}
As eluded to above, the immobility of the topological sectors in the cubic code is reflected in the symmetry-enriched topological order of the compactified models. The relevant symmetries are the translations: $T_x,T_y,$ and $T_z$ which are mapped to translation and on-site symmetries under compactification, respectively. 
The nontrivial action of the translation symmetries on the compactified models leads to symmetry-twisted boundary conditions that remarkably allow the ground space degeneracy of $n_{L_z}$ copies of toric code on an $L_x\times L_y$ torus to match the ground space degeneracy of the cubic code on an $L_x\times L_y \times L_z $ 3D torus. 

The action of $T_x,T_y,$ and $T_z$ on the anyons in compactified cubic code takes the form of a permutation that does not mix $e$ excitations of the $X$ stabilizers with $m$ excitations of the $Z$ stabilizers. 
In fact, due to symmetries of the cubic code Hamiltonian we find that the action of the symmetries simplifies further. The symmetry of reflection across the $\hat{x}\pm \hat{y}$ planes implies ${T}_x$ and ${T}_y$ act on the anyons in an isomorphic way. The further symmetry given by a combined spatial inversion, on-site swap, and Hadamard, implies that the action of any ${T}_i$ on the $e$ sector is isomorphic to the same $T_i$ symmetry acting on the $m$ sector. Hence we focus our attention on the action of ${T}_x$ and ${T}_z$ on the $e$ sector. These actions are specified by $n_{L_z} \times n_{L_z}$ binary matrices $\bar{T}_x$ and $\bar{T}_z$ that describe the action of $T_x$ and $T_z$ on the basis of $e$ string operators found using the approach covered in section~\ref{compactification}. 

While we have explicitly calculated the matrices $\bar{T}_x$ and $\bar{T}_z$ for $L_z\leq 22$, they are only unique up to a change of basis for the anyons and are not enlightening when presented in raw form. For this reason we focus on an important gauge invariant quantity of the symmetry-enriched topological phase: the orders $O_x=O_y$ and $O_z$ of the translation actions within the automorphism group of the anyons, see section~\ref{compactification}. 
The values of $O_i$ are tabulated in table~\ref{table_copies}.

The orders of permutation in the $\hat{x}$ and $\hat{y}$ directions have a clear physical interpretation: $O_x \times O_y$ is the minimal  coarse-grained unit cell for which translation symmetries $T_x^{O_x},T_y^{O_y},$ act trivially on the anyons. In other words, on a torus with ${O_x \not| L_x}$ or ${O_y \not| L_y}$ the boundary conditions are twisted by a nontrivial translation action. For $n_{L_z}$ copies of toric code on an $L_x \times L_y$ torus with twisted boundary conditions, i.e. at least one of ${L_x \mod O_x}$ or ${L_y \mod O_y}$ are nonzero, the ground state degeneracy is reduced to $< 4^{n_{L_z}}$. In fact the ground space degeneracy is known to be given by the number of ${T}_x^{L_x}$-symmetric ${T}_y^{L_y}$ defects, or vice-versa~\cite{barkeshli2014symmetry}. 
Hence, the complicated behavior of the cubic code's ground space degeneracy on a 3D torus can be understood via the symmetry-enriched topological order of the compactified model. 

Observe in table~\ref{table_copies} that the order of translation $O_x$ is doubled as the compactification radius $L_z$ is doubled. 
This, in fact, holds in general which we will proceed to prove. 
To achieve this we make use of the following formula for $\log_{2}$ of the ground space degeneracy of the cubic code on an $L_1 \times L_2 \times L_3$ torus~\cite{Haah2013}
\begin{equation}
\label{eq:ccdegen}
d=\begin{cases}
2^{l_{1}+1}d_{1}-1 & \text{if }l_{1}=l_{2}\, ,\\
2^{l_{1}+1}d_{1} & \text{otherwise.}
\end{cases}
\end{equation}
In the above, without loss of generality,    $l_1\leq l_2 \leq l_3$ indicate the number of factors of 2 in the prime factorizations of $L_1,L_2,L_3,$ and 
\begin{align}
    d_1 = \deg_x {\gcd} \left( (1+ x)^{\ell_1}+1,\right. (&1+ \omega x)^{\ell_2}+1, 
     \nonumber 
    \\
    &\left. (1+ \omega^2 x)^{\ell_3}+1
    \right)_{\mathbb{F}} ,
\end{align}
where $\ell_i=2^{-l_i}L_i$, $1+\omega+\omega^2=0$ and $\mathbb{F}$ is the algebraic closure of $\mathbb{F}_2$. 
In particular, if $L_1=L_2=L_3$ then $l_1=l_2$. 

\begin{lemma}
\label{lem:1}
For any value $L_z = 2^{l_z} \ell_z,$ consider $O:=O_x=O_y=2^o O'$, where ${2 \not | \, \ell_z,  O'}$. Then, $o-l_z=1$. 
\end{lemma}
\begin{proof}
The proof is by contradiction: First let us assume ${o=l_z}$. If we change the length of the  periodic 2D lattice directions from $L_x=L_y=O$, which implies $l_x=l_y=l_z=o$,  to ${L_x}^\prime ={L_y}^\prime = 2^{i}O$ for any natural number $i$, which implies $l_z=o<l_x^\prime=l_y^\prime =o+i$, the degeneracy must change from $2^{o+1} d_1-1$ to $2^{o+1} d_1$.
Since the value of $d_1$ doesn't depend on $l_x,l_y,l_z,$ the degeneracy has decreased upon multiplying $L_x,L_y$ by $2^i$. 
However, the degeneracy should be maximal, and constant, for all untwisted boundary conditions $L_x=L_y=n O$ with $n$ a natural number. Hence ${o\not =l_z}$. 

Next, let us assume $o-l_z < 0$. Then similarly going from $L_x=L_y=O$ with $l_x=l_y=o<l_z$ to ${L_x}^\prime ={L_y}^\prime = {2^{l_z-o+1}} O$ with  $l_z<l_x^\prime =l_y^\prime=l_z+1$ changes the degeneracy from $2^{o+1} d_1-1$ to $2^{l_z+1}d_1$. However, for $L_x'$ and $L_y'$ being any positive integer multiples of the order $O$, the ground state degeneracy should not change. Thus, $o-l_z< 0$ is not consistent. 

Finally, let us assume $o-l_z\geq 2$. If we change lengths from $L_x=L_y=O$ to ${L_x}^\prime ={L_y}^\prime = {2^{-i}}{O}$, for $0 < i < o - l_z$ we go from $l_z<l_x=l_y=o$ to $l_z <l_x^\prime =l_y^\prime=o-i$, while for $i\geq o-l_z$ we go to  $l_x^\prime =l_y^\prime =o-i \leq l_z$ . Thus, for $0<i<o-{l_z}$, the degeneracy doesn't change. 
This is not consistent with $O$ being the order of a nontrivial anyon permuting translation action.
To see this note that $L_x',L_y'<O$ corresponds to twisted boundary conditions that must lead to a strictly smaller ground space degeneracy than the maximal value for $L_x=L_y=O$.

The only consistent possibility remaining is ${o-l_z=1}$. 
\end{proof}
Next, we argue that the order of the translation action doubles upon doubling the compactification radius, $O(2L_z)=2 O(L_z)$. 
Consider compactifying the cubic code with radius $L_z$, the smallest torus corresponding to untwisted boundary conditions is given by $L_x=L_y=O(L_z)$ which leads to a degeneracy $4^{n_{L_z}}$. 
Doubling the compactification radius to $2 L_z$, we have shown above that $n_{2 L_z} = 2 n_{L_z}$. Then periodic boundary conditions $L_x=L_y=2 O(L_z)$ lead to the maximal possible degeneracy $4^{2 n_{L_z}}$ and hence correspond to untwisted boundary conditions. Therefore $2 O(L_z) = k O( 2 L_z)$ for a natural number $k$, and the above lemma implies ${ 2 \not | k} $. 
Furthermore, we find that the ground space degeneracy for compactification radius $L_z$ and $L_x=L_y=O(L_z) / k$ is again maximal, $4^{n_{L_z}}$.
This implies $k=1$, or else there would be a contradiction with $O(L_z)$ being the order of a nontrivial translation action.

\begin{claim}
Suppose $L_z = 2^i \ell$ with $\ell$ odd.
Then, we have $O_x = O_y \le 2^{i+1}(2^k - 1)$
where $k$ is the least common multiple of $2$ and the multiplicative order of $2$ modulo $L_z$. In particular,
\begin{align}
O_x &\le 2^{i+1}( 2^{\ell-1} - 1 ),\\
O_x &= 2^{i+1}\ell &\text{ if } \ell = 2^n - 1 \text{ for $n$ even}, \\
O_x &\le 2^{i+1} \ell(\ell + 2) &\text{ if } \ell = 2^n -1 \text{ for $n$ odd}.
\end{align}
We also have $O_z = L_z$ if $\ell \neq 3$ and $O_z = 2^i$ if $\ell = 3$.
\end{claim}
\begin{proof}
For the above proposition, it suffices to assume $i=0$.
We use the polynomial method~\cite{Eisenbud} and Galois theory~\cite{Lang}.
Recall that the charge module for the compactified cubic code is 
$\FF_2[x^\pm,y^\pm,z]/J$ where $J = (1+x+y+z, 1+xy+yz+zx, z^\ell-1)$.
The translation group acts as the monomial multiplication.
As noted earlier, $O_x = O_y$ from the symmetry $x \leftrightarrow y$.
The order $O_x$ is the minimum positive integer such that
$x^{O_x} -1 \in J$.
We estimate the order $O_x$ by considering the zeros of the ideal $J$.
The first two generators of $J$ defines a variety of codimension~2,
and the compactification condition $z^L -1$ selects finitely many points
in this variety.
Specifically, the first two generators define two lines that are parametrized as
$(\omega^2+ \omega t, \omega + \omega^2 t, t)$ and 
$(\omega+ \omega^2 t, \omega^2+ \omega t, t)$
where $t$ is arbitrary, and $\omega \in \FF_4$ is a third root of unity
satisfying $\omega^2 + \omega + 1 = 0$.
The polynomial $z^\ell -1$ is separable (no degenerate roots),
and the variety will be rational over the minimal extension field $\FF_{2^k}$ over $\FF_2$
that contains all the $\ell$-th root of unity, which form a cyclic group $\ZZ/\ell\ZZ$,
and $\omega$.
Every nonzero element of $\FF_{2^k}$ is a root of $x^{2^k-1} -1$,
a power of which must belong to $J$.
Localizing at the points of the two lines we see that $(x^{2^k-1} -1)^2 \in J$.
(Localization at $(1,1,1)$ shows that $x=1$ has two-fold degeneracy.)
Thus, the order $O_x$ is at most $2\cdot (2^k - 1)$.
It remains to compute $k$ for inequalities.

Since any automorphism of $E$ over $\FF_2$ should send 
a generator $\zeta$ of the group of all $L$-th root of unity
to another generator, the splitting field $\FF_2(\zeta)$ of $z^\ell -1$
has at most $|(\ZZ / \ell\ZZ)^\times| = \varphi(\ell)$ automorphisms over $\FF_2$.
In fact, since the automorphism of a finite field 
of characteristic~2
is always a composition of Frobenius map $\gamma \mapsto \gamma^2$,
the group $\mathrm{Aut}(\FF_2(\zeta) / \FF_2)$
has order $k'$ that is equal to the multiplicative order of 2 modulo $\ell$.
By Artin's theorem we see the extension degree 
$[\FF_2(\zeta):\FF_2]$ equals $k'$,
which divides $\varphi(\ell) \le \ell-1$.
If $k'$ is even, then $\FF_2(\zeta)$ already contains $\omega$ 
(since $\omega$ has degree-two minimal polynomial),
and we have $k' = k \le \ell-1$.
If $k'$ is odd, then $\FF_2(\zeta)$ does not contain $\omega$
and we have $k = 2k'$.
Since $L-1$ is even, in any case we have $k | (\ell-1)$.
If $\ell+1 = 2^n$,
then the multiplicative order of $2$ modulo $\ell$ is $n$.
This proves all the inequalities.

Furthermore, if $\ell+1 = 2^n$ with $n$ even, 
then the $x$-coordinate of the two lines ranges over 
exactly all the nonzero elements of $\FF_{2^n}$,
we see $x^m - 1 \in \mathrm{rad} J$ 
if and only if $m$ is a multiple of $n$.
This proves the lower bound in this special case, hence the equality.

To compute $O_z$, we again look at the two lines.
Since we have to avoid $x=0$ or $y=0$ planes as the variables $x$ and $y$
are invertible, the parameter $t$ is not equal to any of $\omega,\omega,0$, but otherwise any value is allowed.

Localizing at maximal ideals $\mm = (x-a,y-b,z-c)$ where $(a,b,c) \neq (1,1,1)$,
we immediately see that $J_\mm$ contains $(z-c)^m$ for $m = 2^i$
but not for $0\le m < 2^i$,
The same is true for $\mm = (x-1,y-1,z-1)$ thanks to
the Gr\"obner basis computation that we have performed 
in the course of computing $n_{L_z}$ above.
The ``multiplicity'' of any $z$-coordinate is $2^i$.

If $3$ does not divide $\ell$,
then $0,\omega^\pm$ are not roots of $z^\ell - 1$,
so $z^m - 1 \in \mathrm{rad} J$ for $m = \ell$ but not for $m < \ell$.
By the multiplicity of the previous paragraph, $O_z = L_z$.
If $3$ divides $\ell$, then $z^2 + z + 1$ divides $z^\ell - 1$ 
and the $z$-coordinates takes values
precisely among the roots of $f(z) = (z^\ell - 1)/(z^2+z+1)$.
By the multiplicity, $f(z)^{2^i} \in J$ 
and any smaller power than $2^i$ will invalidate the membership.
It remains to find the minimum $m \ge 0$ such that
$f(z)^{2^i} | (z^m - 1)$.
Let $m = 2^j m'$ with $m'$ odd.
For the multiplicity, we know $j \ge i$ and $f(z) | (z^{m'} - 1)$.
The degree of $f(z)$ is $\ell - 2$ which must be $\le m'$,
so minimum $m'$ is one of $\ell -2, \ell -1, \ell$.
If $\ell = 3$, then $m' = 1$ is the minimum.
If $\ell > 3$, then a root of $f(z)$ 
is a primitive $\ell$-th root of unity
which can be a root of $z^{m'}-1$ only if $m' = \ell$.
\end{proof}

\subsubsection{Translation Defects}
As mentioned above, the degeneracy due to twisted boundary conditions on an $L_x \times L_y$ torus is equal to the number of $T_x^{L_x}$-symmetric $T_y^{L_y}$ defects, and vice versa~\cite{barkeshli2014symmetry}. 
 By combining Eq.~\eqref{eq:ntc} with the degeneracy formula in Eq.~\eqref{eq:ccdegen} and Lemma~\ref{lem:1} we can compute a large range of ground state degeneracies and hence numbers of symmetric defects. 
Fix a compactification radius $L_z$, for $L_x=L_y=O$ we have $d_1 = 2^{-l_z} n_{L_z}$. The degeneracy $2^d$ for $L_x=O/2^{i},\, L_y=O/2^{j},$ specifies the number of $T_x^{O/2^{i}}$-symmetric $T_y^{O/2^{j}}$ defects as follows 
\begin{align}
    d=
    \begin{cases}
    2^{l_z+1} d_1 - 1 & \text{ for } i=0,\,j=1,
    \\
    2^{l_z+2-i} d_1 - 1 & \text{ for } i=j>0,
    \\
     2^{l_z+2-i} d_1  & \text{ for } i>j>0,
    \\
     2^{l_z+2-j} d_1  & \text{ for } j>i\geq 0,\, j>1.
    \end{cases}
\end{align}

Since the compactified models are equivalent to copies of toric code, they only support abelian anyons. Hence all the $T_x^{L_x}$ translation defects have the same quantum dimension, as they are related by fusion with the anyons. We recall that the total quantum dimension of any defect sector is equal to that of the anyons~\cite{barkeshli2014symmetry}. This implies that the number of $T_x^{L_x}$ defects, $N_{L_x}$, is related to their common quantum dimension, $\delta_{L_x}$, via $N_{L_x} \delta_{L_x}^2 = 4^{n_{L_z}}$. Since $\log_2 N_{L_x} = d$, where $d$ is $\log_2$ of the ground space degeneracy of the cubic code on an $L_x \times O(L_z) \times L_z$ 3D torus, we have $\log_2 \delta_{L_x} = { n_{L_z} - d/2}$.

Since $N_{L_x}$ is the same as the number of anyon types that are symmetric under a translation $T_x^{L_x}$ we can compute it directly from the dimension of the invariant subspace under $\bar{T}_x^{L_x}$. In particular, we have computed the number of unit translation $T_x$ defects in this way for the values of $L_z$ listed in table~\ref{table_copies}
\begin{align}
    \log_2 N_1 = 
    \begin{cases}
    4 &  \text{ for }  2\ | L_z, 
    \\
    2 &  \text{ for }  2 \not | L_z. 
    \end{cases}
\end{align} 
Hence the quantum dimensions $\delta_1$ fit the following formula as a function of $L_z$ 
\begin{equation}
	\log_2 \delta_1 =
	\begin{cases}
		n_{L_z}-2 &  \text{ for } 2\ | L_z, \\
		n_{L_z}-1 &  \text{ for } 2 \not | L_z .
	\end{cases}
\end{equation}
Our above results for $N_1$ imply that changing the periodic boundary conditions of the cubic code slightly: from $O(L_z) \times O(L_z) \times L_z$ to $\left[ O(L_z) +1\right] \times O(L_z) \times L_z$, can result in an extensive jump of the ground space degeneracy: from $4^{n_{L_z}}$ to $N_1$. 
This provides an appealing interpretation for the seemingly erratic behaviour of the ground space degeneracy of the cubic code on periodic boundary conditions in terms of the better understood phenomena of translation symmetry twisted boundary conditions of a 2D topological order.

\subsection{Cubic code B}

Our second example is cubic code B~\cite{haah2014bifurcation} which is specified by the following generators
\begin{align}
\begin{array}{c}
\drawgenerator{}{}{XXXI}{}{IIXX}{IXIX}{}{XIII}
\quad
\drawgenerator{ZZII}{IZIZ}{}{IIZI}{}{}{ZIZZ}{}
\\
\drawgenerator{}{}{XIIX}{}{IIXI}{XXXX}{}{IXII}
\quad
\drawgenerator{ZIII}{ZZZZ}{}{IIIZ}{}{}{IZZI}{}
\end{array}
\, .
\end{align}
This model was found to bifurcate from the original cubic code during real-space entanglment-renormalization and hence is also type-II. More specifically, a local unitary circuit $U$ was found which is invariant under a coarse-grained translation group, generated by $T_x^2,T_y^2,T_z^2,$ and satisfies
\begin{align}
    U H_A(a) U^\dagger \cong H_A(2a) + H_B(2a) 
    \, .
\end{align}
In the above $H_A(a)$ denotes the original cubic code A Hamiltonian with lattice spacing $a$, $H_B(2a)$ denotes the B cubic code Hamiltonian with lattice spacing $2a$, and $\cong$ denotes that the stabilizer group of two Hamiltonians is the same, up to tensoring with ancilla qubits in the product state $\ket{0}^{\otimes N}$. We remark that the equivalence $\cong$ does not imply the stabilizer generators on the left and right strictly match after the removal of ancillas. 

Cubic code B was further found to be self-bifurcating. That is under another local unitary circuit $V$, which respects the coarse-grained translation group generated by $T_x^2,T_y^2,T_z^2,$ \R{haah2014bifurcation} found
\begin{align}
    V H_B(a) V^\dagger \cong H_B(2a) + H_B'(2a) 
    \, .
\end{align}

For the compactification of cubic code B along the $\hat{z}$ direction  we have found 
\begin{align}
    n^B_{L_z} &= n^A_{L_z} \, , \nonumber \\
    O^B(L_z) &= O^A(L_z)/2 \, , \nonumber \\
    \log_2 N_1^B &= 4 \, .
\end{align}
Hence the cubic codes A and B lead to the same topological phase under compactification. However, this does not imply the original models lie in the same topological phase. 
On the other hand, the translation symmetry-enriched phases after compactification are different, but this is not yet sufficient to show that the original models lie in different phases. 
This is because one is allowed an arbitrary, finite coarse-graining step when comparing the original models. 
An obvious first step is to coarse-grain the $\hat{x}$ and $\hat{y}$ directions of cubic code A by a factor of two, to ensure the order of the translation actions match. After this transformation, the symmetry-enriched topological orders resulting from the compactification of cubic codes A and B are still distinct. To see this we consider the number of unit translation defects for the compactified coarse-grained A code, which equals the number of $T_x^2$ translation defects for the original compactified A code, 
\begin{align}
    \log N_2^A = 
      \begin{cases}
    8 &  \text{ for }   4 \ | L_z,
    \\
    6 &  \text{ for }  2 \ | L_z,\  4 \not | L_z,
    \\
    4 &  \text{ for }  2 \not | L_z. 
    \end{cases}
\end{align}
No coarse-graining in the $\hat{z}$ direction can bring this into agreement with $N_1^B$. However, this is still only a necessary condition for the original cubic codes A and B to lie in distinct phases.

It was shown in \R{haah2014bifurcation} that it is impossible to find any coarse-grainings of cubic codes A and B that make their ground space degeneracies agree for periodic boundary conditions on all system sizes. Hence they are distinct topological phases of matter. 
From the point of view of compactification, this implies that no coarse-grainings can be made such that cubic code A and B lead to the same symmetry-enriched topological orders for all compactifications. 
In particular, since the ground space degeneracies cannot always be made to match, even after an arbitrary coarse-graining the number of symmetric defects of all types in the compactified models cannot be made to match.

\section{ Discussion and conclusions}
\label{sec:conclusion}

In this work we have studied compactifications of translation invariant stabilizer models with fracton topological order. 
We found that type-I fracton order is reflected by a combination of topological order and symmetry breaking in the compactified model that has a simple scaling with the compactification radius. 
More interestingly, we found that type-II fracton order manifests in the symmetry-enriched topological order of the compactified model. The topological order has a relatively simple scaling with the compactification radius, but is enriched in a complicated way by translation symmetry. We have analytically and numerically studied various aspects of the compactification in detail for the 2D and 3D toric code, X-cube model, and the cubic code. Our results on the cubic code provide an understanding of the model's complicated ground space degneracy in terms of twisted boundary conditions for copies of the 2D toric code enriched by translation symmetry. More generally our results draw a connection between fracton topological phases in 3D and translation symmetry-enriched topological phases in 2D which may prove useful for their classification.

Furthermore, we encountered nontrivial subsystem symmetry-enriched phases in our study of the compactified cubic code, which led to spurious contributions to the topological entanglement entropy~\cite{KitaevPreskill,levinwenentanglement}, see \R{Williamson2018} for further details. The possibility of such spurious contributions in two and three dimensions cause a complication for proposals to extract information about fracton topological orders from the scaling behaviour of the entanglement entropy~\cite{PhysRevB.97.144106,PhysRevB.97.125101,PhysRevB.97.125102}. 

It would be interesting to explore the implications of compactification for decoding fracton topological codes~\cite{PhysRevLett.107.150504,bravyi2013quantum}. While this may not be particularly useful for type-I models~\cite{Brown2019}, since type-II models remain topological codes under compactification one may be able to apply techniques based on the 2D toric code. In particular it should be possible to correct errors that are nonlocal in one spatial direction. 

Another apparent future direction is the study of compactification with different boundary conditions. Open boundary conditions are particularly appealing due to their practical relevance. It would also be interesting to study the anyons and symmetry defects found after compactification as linelike objects in the uncompactified model. We plan to explore these directions in future work.

\subsection{Further examples of compactification}
The appendix contains numerical calculations of $n_{L_z}$, the number of copies of toric code, for a wide range of fracton stabilizer models that have been compactified in several different directions. It is clear that the scaling of $n_{L_z}$ for TQFT, foliated fracton, fractal and type-II fracton, topological stabilizer models is qualitatively different. 
The sorting of these 3D topological stabilizer is discussed in detail in a forthcoming work~\cite{Dua_2019_FTO}. 
Our results demonstrate that compactification can serve as a useful tool in the sorting and classification of fracton topological orders. 
In \R{Dua_2019_FTO} we go beyond this to discuss the coarse sorting and classification of 3D topological stabilizer models using tools such as the deformability of logical operators and anti-commuting logical operator pairs supported on different configurations. 

\acknowledgments
AD thanks Mengzhen Zhang and Guanyu Zhu for useful discussions. This work is supported by start-up funds at Yale University (DW and MC) and the Alfred P. Sloan foundation (MC). 

\bibliographystyle{apsrev_1}
\bibliography{DomBib}

\newpage
\begin{widetext}
\appendix
\section{Compactification of further examples}

In this appendix we report our numerical calculations of the number of copies of toric code as a function of compactification radius along each axis. We consider a wide range of fracton stabilizer models. The number of toric codes is reported for the compactified model after all local symmetry breaking degeneracy has been lifted by adding 2D local operators to the compactified Hamiltonian. 
The models are labelled CC0-17, standing for cubic code 0-17 following the notation in \R{Haah2013} which differs slightly from \R{haah2011local}. 
We also consider the 3D toric code, labeled 3D TC, Chamon's model~\cite{chamon2005quantum,bravyi2011topological}, labeled Chm, another model found by Castelnovo and Chamon~\cite{doi:10.1080/14786435.2011.609152} and also Yoshida~\cite{yoshida2013exotic}, which we label Y, the X-cube model and checkerboard model~\cite{vijay2016fracton}, labeled XC and CB respectively, and finally the so-called type-I and II spin models in \R{PhysRevB.96.165105}, labeled H-I and H-II respectively. 
We remark that the checkerboard model is local unitary equivalent to two copies of the X-cube model~\cite{shirley2018Foliated} and similarly the H-I model is mapped to two copies of the checkerboard model by applying swap gates to half the sites. We also point out that H-II has \emph{not} been shown to be a type-II model. 

In the tables below we group the fracton models that haven't (have) been previously shown to support a string operator, i.e. the first table contains the type-II models, and some that may be type-I, while the second contains TQFT and type-I models and so on. We show our results for compactification in the $\hat{z},\hat{y},$ and $\hat{x}$ directions, respectively. The quantity $L_i$ denotes the compactification radius, and the superscript $*$ denotes that CB* and H-I* have been coarse grained by a factor of two. 

We observe that the scaling behaviour of the number of toric codes in the compactified model is indicative of the model's type. For TQFT (3D TC) the number of 2D toric codes is constant. For foliated type-I fracton orders (Chm, XC, CB, H-I) the number scales linearly. For type-II and fractal type-I~\cite{Dua_2019_FTO} (CC0-17, H-II, Y) there is at least one direction where the number fluctuates as it grows with the compactification radius. 

\begin{table}[h]
    \centering
\begin{tabular}{|c|c|c|c|c|c|c|c|c|c|c|c|}
\hline 
$L_{z}$ & CC1 & CC2 & CC3 & CC4 & CC5 & CC6 & CC7 & CC8 & CC9 & CC10 & H-II\tabularnewline
\hline 
\hline 
2 & 4 & 2  & 4 & 4 & 2 & 2 & 2 & 2 & 4 & 2  & 8\tabularnewline
\hline 
3 & 2 & 1 & 2 & 4 & 1 & 1 & 5 & 5 & 6 & 1 & 4\tabularnewline
\hline 
4 & 8 & 4 & 8 & 8 & 4 & 4 & 4 & 4 & 8 & 4 & 16\tabularnewline
\hline 
5 & 10 & 9 & 10 & 10 & 5 & 5 & 9 & 9 & 10 & 9 & 20\tabularnewline
\hline 
6 & 4 & 2 & 4 & 8 & 2 & 2 & 10 & 10 & 12 & 2 & 16\tabularnewline
\hline 
7 & 14 & 13 & 14 & 14 & 7 & 7 & 13 & 13 & 14 & 13 & 28\tabularnewline
\hline 
8 & 16 & 8 & 16 & 16 & 8 & 8 & 8 & 8 & 16 & 8 & 32\tabularnewline
\hline 
9 & 14 & 13 & 14 & 16 & 7 & 7 & 17 & 17 & 18 & 13 & 28\tabularnewline
\hline 
10 & 20 & 18 & 20 & 20 & 10 & 10 & 18 & 18 & 20 & 18 & 40\tabularnewline
\hline 
\end{tabular}
\caption{Number of copies of toric code as a function of compactification radius along $\hat{z}$.}
\end{table}

\begin{table}[h]
    \centering
\begin{tabular}{|c|c|c|c|c|c|c|c|c|c|c|c|c|c|c|}
\hline 
$L_{z}$ & 3D TC & Chm & X-cube & CB{*} & H-I{*} & CC0 & CC11 & CC12 & CC13 & CC14 & CC15 & CC16 & CC17 & Y\tabularnewline
\hline 
\hline 
2 & 1 & 2 & 1 & 2 & 4 & 8 & 4 & 4 & 4 & 2 & 4 & 2 & 4 & 0\tabularnewline
\hline 
3 & 1 & 4 & 2 & 4 & 8 & 12 & 4 & 6 & 2 & 6 & 6 & 5 & 2 & 0\tabularnewline
\hline 
4 & 1 & 6 & 3 & 6 & 12 & 16 & 8 & 8 & 8 & 3 & 8 & 4 & 8 & 0\tabularnewline
\hline 
5 & 1 & 8 & 4 & 8 & 16 & 20 & 10 & 10 & 10 & 10 & 10 & 9 & 10 & 0\tabularnewline
\hline 
6 & 1 & 10 & 5 & 10 & 20 & 24 & 8 & 12 & 4 & 10 & 12 & 10 & 8 & 0\tabularnewline
\hline 
7 & 1 & 12 & 6 & 12 & 24 & 16 & 14 & 14 & 14 & 14 & 14 & 13 & 14 & 0\tabularnewline
\hline 
8 & 1 & 14 & 7 & 14 & 28 & 32 & 16 & 16 & 16 & 11 & 16 & 8 & 16 & 0\tabularnewline
\hline 
9 & 1 & 16 & 8 & 16 & 32 & 36 & 16 & 18 & 14 & 18 & 18 & 17 & 14 & 0\tabularnewline
\hline 
10 & 1 & 18 & 9 & 18 & 36 & 40 & 20 & 20 & 20 & 18 & 20 & 18 & 20 & 0\tabularnewline
\hline 
\end{tabular}
\caption{Number of copies of toric code as a function of compactification radius along $\hat{z}$.}
\end{table}

\begin{table}[h]
    \centering
\begin{tabular}{|c|c|c|c|c|c|c|c|c|c|c|c|}
\hline 
$L_{y}$ & CC1 & CC2 & CC3 & CC4 & CC5 & CC6 & CC7 & CC8 & CC9 & CC10 & H-II\tabularnewline
\hline 
\hline 
2 & 4 & 2 & 2 & 2 & 2 & 2 & 4 & 2 & 2 & 2 & 8\tabularnewline
\hline 
3 & 2 & 1 & 1 & 5 & 5 & 5 & 6 & 5 & 5 & 1 & 4\tabularnewline
\hline 
4 & 8 & 4 & 4 & 4 & 4 & 4 & 8 & 4 & 4 & 4 & 16\tabularnewline
\hline 
5 & 10 & 9 & 9 & 9 & 9 & 9 & 10 & 9 & 9 & 9 & 20\tabularnewline
\hline 
6 & 4 & 2 & 2 & 10 & 10 & 10 & 12 & 10 & 10 & 2 & 16\tabularnewline
\hline 
7 & 14 & 13 & 13 & 13 & 13 & 13 & 14 & 13 & 13 & 13 & 28\tabularnewline
\hline 
8 & 16 & 8 & 8 & 8 & 8 & 8 & 16 & 8 & 8 & 8 & 32\tabularnewline
\hline 
9 & 14 & 13 & 13 & 17 & 17 & 17 & 18 & 17 & 17 & 13 & 28\tabularnewline
\hline 
10 & 20 & 18 & 18 & 18 & 18 & 18 & 20 & 18 & 18 & 18 & 40\tabularnewline
\hline 
\end{tabular}
\caption{Number of copies of toric code as a function of compactification radius along $\hat{y}$.}
\end{table}

\begin{table}[h]
    \centering
\begin{tabular}{|c|c|c|c|c|c|c|c|c|c|c|c|c|c|c|}
\hline 
$L_{y}$ & 3D TC & Chm & X-cube & CB{*} & H-I{*} & CC0 & CC11 & CC12 & CC13 & CC14 & CC15 & CC16 & CC17 & Y\tabularnewline
\hline 
\hline 
2 & 1 & 2 & 1 & 2 & 4 & 8 & 4 & 2 & 4 & 2 & 4 & 2 & 4 & 0\tabularnewline
\hline 
3 & 1 & 4 & 2 & 4 & 8 & 12 & 4 & 5 & 4 & 3 & 6 & 5 & 2 & 2\tabularnewline
\hline 
4 & 1 & 6 & 3 & 6 & 12 & 16 & 8 & 4 & 8 & 4 & 8 & 4 & 8 & 0\tabularnewline
\hline 
5 & 1 & 8 & 4 & 8 & 16 & 20 & 10 & 9 & 10 & 9 & 10 & 9 & 10 & 4\tabularnewline
\hline 
6 & 1 & 10 & 5 & 10 & 20 & 24 & 8 & 10 & 8 & 6 & 12 & 10 & 4 & 4\tabularnewline
\hline 
7 & 1 & 12 & 6 & 12 & 24 & 16 & 14 & 13 & 14 & 13 & 14 & 13 & 14 & 6\tabularnewline
\hline 
8 & 1 & 14 & 7 & 14 & 28 & 32 & 16 & 8 & 16 & 8 & 16 & 8 & 16 & 3\tabularnewline
\hline 
9 & 1 & 16 & 8 & 16 & 32 & 36 & 16 & 17 & 16 & 15 & 18 & 17 & 14 & 8\tabularnewline
\hline 
10 & 1 & 18 & 9 & 18 & 36 & 40 & 20 & 18 & 20 & 18 & 20 & 18 & 20 & 8\tabularnewline
\hline 
\end{tabular}
\caption{Number of copies of toric code as a function of compactification radius along $\hat{y}$.}
\end{table}

\begin{table}[h]
    \centering
\begin{tabular}{|c|c|c|c|c|c|c|c|c|c|c|c|}
\hline 
$L_{x}$ & CC1 & CC2 & CC3 & CC4 & CC5 & CC6 & CC7 & CC8 & CC9 & CC10 & H-II\tabularnewline
\hline 
\hline 
2 & 4 & 4 & 2 & 2 & 2 & 4 & 2 & 2 & 2 & 2 & 8\tabularnewline
\hline 
3 & 2 & 2 & 1 & 3 & 5 & 2 & 5 & 5 & 5 & 1 & 4\tabularnewline
\hline 
4 & 8 & 8 & 4 & 4 & 4 & 8 & 4 & 4 & 4 & 4 & 16\tabularnewline
\hline 
5 & 10 & 10 & 9 & 9 & 9 & 10 & 9 & 9 & 9 & 9 & 20\tabularnewline
\hline 
6 & 4 & 4 & 2 & 6 & 10 & 4 & 10 & 10 & 10 & 2 & 16\tabularnewline
\hline 
7 & 14 & 14 & 13 & 13 & 13 & 14 & 13 & 13 & 13 & 13 & 28\tabularnewline
\hline 
8 & 16 & 16 & 8 & 8 & 8 & 16 & 8 & 8 & 8 & 8 & 32\tabularnewline
\hline 
9 & 14 & 14 & 13 & 15 & 17 & 14 & 17 & 17 & 17 & 13 & 28\tabularnewline
\hline 
10 & 20 & 20 & 18 & 18 & 18 & 20 & 18 & 18 & 18 & 18 & 40\tabularnewline
\hline 
\end{tabular}
\caption{Number of copies of toric code as a function of compactification radius along $\hat{x}$.}
\end{table}

\begin{table}[h]
\centering
\begin{tabular}{|c|c|c|c|c|c|c|c|c|c|c|c|c|c|c|}
\hline 
$L_{x}$ & 3D TC & Chm & X-cube & CB{*} & H-I{*} & CC0 & CC11 & CC12 & CC13 & CC14 & CC15 & CC16 & CC17 & Y\tabularnewline
\hline 
\hline 
2 & 1 & 2 & 1 & 2 & 4 & 8 & 2 & 4 & 4 & 2 & 2 & 2 & 4 & 0\tabularnewline
\hline 
3 & 1 & 4 & 2 & 4 & 8 & 12 & 1 & 6 & 6 & 2 & 5 & 5 & 2 & 2\tabularnewline
\hline 
4 & 1 & 6 & 3 & 6 & 12 & 16 & 4 & 8 & 8 & 3 & 4 & 4 & 8 & 0\tabularnewline
\hline 
5 & 1 & 8 & 4 & 8 & 16 & 20 & 5 & 10 & 10 & 10 & 9 & 9 & 10 & 4\tabularnewline
\hline 
6 & 1 & 10 & 5 & 10 & 20 & 24 & 4 & 12 & 12 & 2 & 10 & 10 & 4 & 4\tabularnewline
\hline 
7 & 1 & 12 & 6 & 12 & 24 & 16 & 7 & 14 & 14 & 14 & 13 & 13 & 14 & 6\tabularnewline
\hline 
8 & 1 & 14 & 7 & 14 & 28 & 32 & 8 & 16 & 16 & 11 & 8 & 8 & 16 & 3\tabularnewline
\hline 
9 & 1 & 16 & 8 & 16 & 32 & 36 & 7 & 18 & 18 & 14 & 17 & 17 & 14 & 8\tabularnewline
\hline 
10 & 1 & 18 & 9 & 18 & 36 & 40 & 10 & 20 & 20 & 18 & 18 & 18 & 20 & 8\tabularnewline
\hline 
\end{tabular}
\caption{Number of copies of toric code as a function of compactification radius along $\hat{x}$.}
\end{table}

\end{widetext}
\end{document}